\documentclass[12pt]{amsart}
\usepackage{amssymb}
\usepackage{color}
\usepackage{amsmath,epic,curves}
\usepackage[english]{babel}
\usepackage{graphicx}
\usepackage{comment}
\newtheorem{claim}{}[section]
\newtheorem{theorem}[claim]{Theorem}

\newtheorem{lemma}[claim]{Lemma}
\newtheorem{proposition}[claim]{Proposition}

\newtheorem{example}[claim]{Example}

\renewenvironment{proof}{\noindent{\it Proof. \hskip0pt}}
                      {$\square$\par\medskip}

\begin{document}
\baselineskip 6.0 truemm
\parindent 1.5 true pc

\newcommand\lan{\langle}
\newcommand\ran{\rangle}
\newcommand\tr{{\text{\rm Tr}}\,}
\newcommand\ot{\otimes}
\newcommand\ol{\overline}
\newcommand\join{\vee}
\newcommand\meet{\wedge}
\renewcommand\ker{{\text{\rm Ker}}\,}
\newcommand\image{{\text{\rm Im}}\,}
\newcommand\id{{\text{\rm id}}}
\newcommand\tp{{\text{\rm tp}}}
\newcommand\pr{\prime}
\newcommand\e{\epsilon}
\newcommand\la{\lambda}
\newcommand\inte{{\text{\rm int}}\,}
\newcommand\ttt{{\text{\rm t}}}
\newcommand\spa{{\text{\rm span}}\,}
\newcommand\conv{{\text{\rm conv}}\,}
\newcommand\rank{\ {\text{\rm rank of}}\ }
\newcommand\re{{\text{\rm Re}}\,}
\newcommand\ppt{\mathbb T}
\newcommand\rk{{\text{\rm rank}}\,}

\newcommand{\bra}[1]{\langle{#1}|}
\newcommand{\ket}[1]{|{#1}\rangle}

\title{Multi-partite separable states with unique decompositions
and construction of three qubit entanglement with positive partial transpose}

\author{Kil-Chan Ha and Seung-Hyeok Kye}
\address{Faculty of Mathematics and Statistics, Sejong University, Seoul 143-747, Korea}
\address{Department of Mathematics and Institute of Mathematics, Seoul National University, Seoul 151-742, Korea}

\thanks{KCH is partially supported by NRFK 2013-020897. SHK is partially supported by NRFK 2009-0083521.}



\begin{abstract}
We investigate conditions on a finite set of multi-partite product
vectors for which separable states with corresponding product states
have unique decomposition, and show that this is true in most cases
if the number of product vectors is sufficiently small. In the three
qubit case, generic five dimensional spaces give rise to faces of
the convex set consisting of all separable states, which are
affinely isomorphic to the five dimensional simplex with six
vertices. As a byproduct, we construct three qubit entangled PPT
edge states of rank four with explicit formulae. This covers those
entanglement which cannot be constructed from unextendible product basis.
\end{abstract}

\maketitle

\section{Introduction}

Theory of entanglement arising from quantum physics is now one of
the main research topics in information sciences,
and involves various fields of mathematics such as algebraic geometry, matrix theory, functional analysis
and combinatorics. Since a quantum state which is not separable is entangled,
it is important to understand the structures of separable states. We note that the set $\mathbb S$ of all separable states is a convex set, and so
it is urgent to characterize the facial structures of the convex set $\mathbb S$.

A state on the Hilbert space ${\mathcal H}=\mathbb C^{d_1}\ot \mathbb C^{d_2}\ot\cdots\ot\mathbb C^{d_n}$
is said to be separable if it is the convex combination
$$
\varrho:=\sum_{i\in I}p_i |z_i\rangle\langle z_i|
$$
of pure product states $|z_i\rangle\langle z_i|$, where
\begin{equation}\label{pv-form}
|z_i\rangle=|x_{1i}\rangle\ot |x_{2i}\rangle\ot\cdots\ot |x_{ni}\rangle\in \mathbb C^{d_1}\ot \mathbb C^{d_2}\ot\cdots\ot\mathbb C^{d_n},\qquad i\in I.
\end{equation}
Therefore, $\varrho$ is a $d\times d$ matrix, with the dimension $d=\prod_{i=1}^nd_i$ of the Hilbert space $\mathcal H$.
By definition, pure product states are extreme points of the convex set $\mathbb S$, which are $0$-dimensional simplices.
Therefore, the first step to understand the facial structures of $\mathbb S$ is to
search for faces which are affinely isomorphic to higher dimensional simplices, as it was initiated
by Alfsen and Shultz \cite{alfsen,alfsen_2} for bi-partite case with $n=2$. We call those simplicial faces.

We note that a point of a convex set determines a unique face in which it is an interior point.
It is clear that the separable state $\varrho$ determines a simplicial face of $\mathbb S$
if and only if $\varrho$ has a unique decomposition,
for which we have a simple sufficient condition: If $\{|z_i\rangle:i\in I\}$ is linearly independent
and its span has no more product vectors then it is clear that $\varrho$ has a unique decomposition. This condition has been considered by
several authors for bi-partite cases. See \cite{chen_dj_ext_PPT,cohen,kirk} for example, as well as \cite{alfsen,alfsen_2}.
The authors \cite{ha-kye-sep-face,ha-kye-2x4} utilized the fact that linear independence of
states $\{|z_i\rangle\langle z_i|:i\in I\}$ are sufficient
for this purpose, and found simplicial faces with higher dimensions for bi-qutrit and qubit-qudit cases.
This approach was also very useful to construct entangled states with positive partial transposes.
Searching for those entanglement is of independent interest in the contexts of PPT criterion \cite{choi-ppt,peres} and
distillation problem \cite{horo-distill}.

The purpose of this paper is to continue this line of research for general multi-partite cases.
We note that arbitrary choice of finitely many product vectors (\ref{pv-form})
gives those in general position with the probability one, that is, any choice of $d_j$ vectors from $j$-component
gives linearly independent vectors. Suppose that we have $k$ product vectors in general position. If
$k\le \sum_{i=1}^n (d_i-1)+1$ then we show that they are linearly independent, and if $k\le \sum_{i=1}^n (d_i-1)$
then there is no more product vectors in their span up to scalar multiplications. Therefore, the corresponding
product states make a simplicial face. This extends the result in \cite{chen_dj_ext_PPT} to multi-partite case.
In the $n$-qubit case, we see that an arbitrary choice of $n$ product vectors in general position gives us a separable state
with unique decomposition.

In the three-qubit case, we go further. We show that four product
vectors in general position give separable states with unique
decomposition if and only if they are also in general position as
product vectors in $\mathbb C^2\otimes \mathbb C^4$. Those separable
states are of rank four. As for rank five separable states, we note
that generic five dimensional spaces have exactly six product
vectors.
It turns out that these six product vectors give rise to linearly
independent pure product states, and any choice of five product vectors among them
must be linearly independent.

Therefore, they give rise to a simplicial face which
is affinely isomorphic to the five dimensional simplex $\Delta_5$
with six extreme points.
If we take an interior point $\varrho_c$
of this face then it is of rank five.
Because an interior point $\varrho_1$ in a maximal face also has rank five,
we can extend the line segment from $\varrho_c$ to $\varrho_1$ to get an entangled state with positive partial transpose (PPT).
The endpoint of this
line segment in the set of PPT states must be a PPT entangled state (PPTES) of rank four.
A standard method to construct a PPTES is to use unextendible product basis \cite{bdmsst}.
This method gives rise to general construction of a PPTES of rank four in the two qutrit case, as it was shown in \cite{chen,sko}.
But, it is far from being true in three qubit case, since generic four dimensional spaces have no product vectors.
Our construction covers PPT entangled states of rank four which cannot be constructed by unextendible product basis.
Furthermore, we give an explicit formula for those entangled states in terms of six product vectors
in generic five dimensional spaces.

In the next section, we collect basic facts on product vectors in general position and generalized unextendible product basis
and give examples for further references, even though some of them must be known to the specialists.
We consider the general multi-partite cases in Section 3, and concentrate on three qubit case in Section 4.
After we construct three qubit PPT entangled states with rank four in Section 5,
we close this paper with discussions on $n$-qubit cases and questions in the last section.
Throughout this paper,  a product vector $|x_1\rangle \ot |x_2\rangle \ot \cdots \ot |x_n\rangle$ will also be written as
$|x_1,x_2,\ldots,x_n\rangle$, and we denote by $\{|e_1\rangle,\dots,|e_d\rangle\}$  the usual basis of $\mathbb C^d$.

We are grateful to Lin Chen for useful discussion on the topics.
\section{General position and generalized unextendible product basis}\label{gp_gupb}

A finite set $\{|z_i\rangle:i\in I\}$ of product vectors in (\ref{pv-form})
is said to be in general position (GP), if  for each $j=1,2,\dots,n$ and a subset $J$ of $I$ with $|J|\le d_j$
the set $\{|x_{ji}\rangle: i\in J\}$ is linearly independent in $\mathbb C^{d_j}$, where $|J|$ denotes
the cardinality of the set $J$. On the other hand, it is called
a generalized unextendible product basis (GUPB) if the orthogonal complement of $\spa\{|z_i\rangle:i\in I\}$ has no
product vector in $\mathcal H$. It is easy to check if a given set of product vectors is a generalized unextendible product basis or not,
by the following proposition in \cite{sko}, Proposition 2.10. We include a simple proof for the convenience of the readers.

\begin{proposition}\label{gupb}
A set $\{|z_i\rangle:i\in I\}$ of product vectors in $\mathcal H$ is a generalized unextendible product basis if and only if the following
is satisfied: For any partition $I_1\cup I_2\cdots \cup I_n$ of $I$, there exists $j\in\{1,\dots,n\}$ such that the set
$\{|x_{ji}\rangle: i\in I_j\}$ spans $\mathbb C^{d_j}$.
\end{proposition}

\begin{proof}
Let $|z_i\rangle$ be of the form $|z_i\rangle=|x_{1i},x_{2i},\ldots, x_{ni}\rangle$ as in \eqref{pv-form}.
For the \lq only if\rq\ part, suppose that there exists a partition $I_1\cup\cdots \cup I_n$ of $I$ such that
$\{|x_{ji}\rangle:i\in I_j\}$ does not span
$\mathbb C^{d_j}$ for any $j=1,2,\dots,n$. We take $|y_j\rangle\in \mathbb C^{d_j}$ orthogonal to the span of $\{|x_{ji}\rangle:i\in I_j\}$
for each $j=1,\dots,n$, then
we see that $|y_1,\,\ldots,\,y_n\rangle$ is orthogonal to the span of $\{|z_i\rangle:i\in I\}$.

For the \lq if\rq\ part, assume  that
there exists a product vector $|y_1,\ldots, y_n\rangle$ which is orthogonal to $|z_i\rangle$ for each $i\in I$, and put
$$
I_j=\{i\in I: \langle x_{ji}| y_j\rangle=0\},\qquad j=1,2,\dots,n.
$$
The assumption implies $I=I_1\cup\cdots\cup I_n$. We can take a nonempty subset $\hat I_j$ of $I_j$
for each $j=1,\dots,n$ so that
$I=\hat I_1\cup\cdots\cup \hat I_n$ is a partition of $I$. It is clear that
no $\{|x_{ji}\rangle: i\in \hat I_j\}$ spans $\mathbb C^{d_j}$.
\end{proof}

Therefore, the minimum number $|I|$ of product vectors to be a generalized unextendible product basis is given by
$\sum_{i=1}^n(d_i-1)+1$. In fact, it is known \cite{walgate} that the maximum dimension of subspaces without product vectors is given by
\begin{equation}\label{s_max}
s_{\max}=\prod_{i=1}^n d_i-\left(\sum_{i=1}^n(d_i-1)+1\right),
\end{equation}
and generic $s_{\max}+1$ dimensional spaces have ${\bigl(\sum_i(d_i-1)\bigr)!}/{\bigl(\prod_i(d_i-1)!\bigr)}$ product vectors. See also
\cite{juhan}. If $|I|\ge \sum_{i=1}^n(d_i-1)+1$, then it is clear that product vectors in general position
must be a generalized unextendible product basis. The converse also holds if
$|I|=\sum_{i=1}^n(d_i-1)+1$, as it was observed in \cite{sko}, Proposition 2.4, for the bi-partite case.

\begin{proposition}\label{gp}
A set $\{|z_i\rangle:i\in I\}$  of product vectors in $\mathcal H$ with $|I|=\sum_{i=1}^n(d_i-1)+1$ is in general position if and only if
it is a generalized unextendible product basis.
\end{proposition}

\begin{proof}
It remains to prove the \lq if\rq\ part.
Suppose that $J\subset I$ with $|J|\le d_{j}$. Then we can take a partition $I_1\cup\cdots \cup I_n$ of $I$
satisfying
$$
I_j\supset J,\qquad |I_j|=d_j,\qquad |I_\ell|=d_\ell-1\ \ \text{\rm for}\ \ell=1,\dots,j-1,j+1,\dots,n.
$$
Then proposition \ref{gupb} implies that $\{|x_{ji}\rangle:i\in I_j\}$ span $\mathbb C^{d_j}$, and so we conclude that
$\{|x_{ji}\rangle:i\in J\}$ is linearly independent in $\mathbb C^{d_j}$.
\end{proof}

\newcommand\zz{\,\,}

\begin{example}\label{exam-a}
Recall that a generalized unextendible product basis is said to be just an unextendible product basis if they are orthogonal.
An example of three qubit unextendible product basis is given in \cite{bdmsst}:
$$
\begin{array}{rrrrrrrrrl}
|z_1\rangle=( \hspace{-0.3truecm}& \cdot\ ,  \hspace{-0.2truecm}& \cdot\ , \hspace{-0.2truecm}& +, \hspace{-0.2truecm}&
 +,\hspace{-0.2truecm}& \cdot\ , \hspace{-0.2truecm}& \cdot\ , \hspace{-0.2truecm}&
  \cdot\ , \hspace{-0.2truecm}& \cdot \ & \hspace{-0.3truecm})^{\rm t}=|e_1\rangle\ot |e_2\rangle\ot (|e_1\rangle+|e_2\rangle)\\
|z_2\rangle=( \hspace{-0.3truecm}& \cdot\ , \hspace{-0.2truecm}& \cdot\ , \hspace{-0.2truecm}&
 \cdot\ , \hspace{-0.2truecm}& \cdot\ , \hspace{-0.2truecm}& +, \hspace{-0.2truecm}&
  \cdot\ , \hspace{-0.2truecm}& +, \hspace{-0.2truecm}& \cdot \ & \hspace{-0.3truecm})^{\rm t}=|e_2\rangle\ot (|e_1\rangle+|e_2\rangle)\ot |e_1\rangle\\
|z_3\rangle=( \hspace{-0.3truecm}& \cdot\ , \hspace{-0.2truecm}& +, \hspace{-0.2truecm}&
 \cdot\ , \hspace{-0.2truecm}& \cdot\ , \hspace{-0.2truecm}& \cdot\ , \hspace{-0.2truecm}&
  +, \hspace{-0.2truecm}& \cdot\ , \hspace{-0.2truecm}& \cdot\  & \hspace{-0.3truecm})^{\rm t}=(|e_1\rangle+|e_2\rangle)\ot |e_1\rangle\ot |e_2\rangle\\
|z_4\rangle=( \hspace{-0.3truecm} & +\,, \hspace{-0.2truecm}& -, \hspace{-0.2truecm}&
 -, \hspace{-0.2truecm}& +, \hspace{-0.2truecm}& -, \hspace{-0.2truecm}&
  +,\hspace{-0.2truecm} & +, \hspace{-0.2truecm}& - & \hspace{-0.3truecm})^{\rm t}
  =(|e_1\rangle-|e_2\rangle)\ot (|e_1\rangle-|e_2\rangle)\ot (|e_1\rangle-|e_2\rangle),\\
\end{array}
$$
where, $+,-,\cdot$ denote $+1, -1, 0$, respectively.
We also identify $\mathbb C^2\ot\mathbb C^2\ot\mathbb C^2$ with $\mathbb C^8$
with the basis in the lexicographic order:
$$
\begin{aligned}
|e_1,\, e_1,\, e_1\rangle,\quad
&|e_1,\, e_1,\, e_2\rangle,\quad
|e_1,\, e_2,\, e_1\rangle,\quad
|e_1,\, e_2,\, e_2\rangle,\\
&|e_2,\, e_1,\, e_1\rangle,\quad
|e_2,\, e_1,\, e_2\rangle,\quad
|e_2,\, e_2,\, e_1\rangle,\quad
|e_2,\, e_2,\, e_2\rangle.
\end{aligned}
$$
Note that the orthogonal complement is spanned by
$$
\begin{array}{lrrrrrrr}
|w_1\rangle=(\ \, \cdot\ ,&\cdot\ ,&1\, ,\hspace{-0.2truecm}&-1\, ,
 &\cdot\ ,\hspace{-0.2truecm}&\cdot\ ,\hspace{-0.2truecm}&\cdot\ ,\hspace{-0.2truecm}&-2\, )^{\rm t}\\
|w_2\rangle=(\ \, \cdot\ ,&\cdot\ ,&\cdot\ ,\hspace{-0.2truecm}
 &\cdot\ , &1\, ,\hspace{-0.2truecm} &\cdot\ ,\hspace{-0.2truecm}&-1\, ,\hspace{-0.2truecm} &-2\, )^{\rm t}\\
|w_3\rangle=(\ \, \cdot\ , &1\, , &\cdot\ ,\hspace{-0.2truecm}
 &\cdot\ ,  &\cdot\ ,\hspace{-0.2truecm} &-1,\hspace{-0.2truecm} &\cdot\ ,\hspace{-0.2truecm} &-2\, )^{\rm t}\\
|w_4\rangle=(\ \, 1\, , &\cdot\ ,&\cdot\ ,\hspace{-0.2truecm}
 &\cdot\ ,&\cdot\ ,\hspace{-0.2truecm}&\cdot\ ,\hspace{-0.2truecm}&\cdot\ ,\hspace{-0.2truecm} &1\, )^{\rm t}.
\end{array}
$$
The orthogonal complement of the three vectors $\{|w_1\rangle,\, |w_2\rangle,\, |w_3\rangle\}$ has exactly six product vectors
$\{|z_1\rangle,\dots, |z_5\rangle,|z_6\rangle\}$, where
$$
\begin{array}{rcccccccl}
|z_5\rangle=(\ 8\, ,&\hspace{-0.2truecm}4\, ,&\hspace{-0.2truecm}4\, ,&\hspace{-0.2truecm}2\, ,&\hspace{-0.2truecm}4\, ,&\hspace{-0.2truecm}2\, ,&\hspace{-0.2truecm}2\, ,&\hspace{-0.2truecm}1\ )&\hspace{-0.2truecm}=(2|e_1\rangle+|e_2\rangle)\ot (2|e_1\rangle+|e_2\rangle)\ot (2|e_1\rangle+|e_2\rangle),\\
|z_6\rangle=(\ 1\, ,&\hspace{-0.2truecm}\cdot\hspace{0.1truecm} ,&\hspace{-0.2truecm}\cdot\hspace{0.1truecm} ,&\hspace{-0.2truecm}\cdot\hspace{0.1truecm},&\hspace{-0.2truecm}\cdot\hspace{0.1truecm},
&\hspace{-0.2truecm}\cdot\hspace{0.1truecm},&\hspace{-0.2truecm}\cdot\hspace{0.1truecm},
&\hspace{-0.2truecm}\cdot\hspace{0.1truecm}\ )&\hspace{-0.2truecm}=|e_1\rangle\ot |e_1\rangle\ot |e_1\rangle.
\end{array}
$$
We summarize properties for a subset $\mathcal S$ of $\{|z_1\rangle,\dots,|z_6\rangle\}$ as follows:

\medskip
\begin{center}
\begin{tabular}{|c|c|c|c|}
  \hline
$|\mathcal S|$ & $|z_6\rangle\in \mathcal S$ & GP & GUBP\\
  \hline
  \hline
6 & Yes & No & Yes\\
\hline
5 & No & Yes & Yes\\
\hline
5& Yes&No&Yes\\
\hline
4&No&Yes&Yes\\
\hline
4&Yes&No&No\\
\hline
\end{tabular}
\end{center}
\end{example}
\medskip

\section{Linear Independence of Product Vectors and Product States}

In this section, we consider $k$ product vectors in (\ref{pv-form}) in general position, and
deal with the question to what extent they make separable states with unique decomposition.
We begin with the question of linear independence of product vectors themselves.

\begin{proposition}\label{prop1}
Suppose that $k$ product vectors in $\mathbb C^{d_1}\ot \mathbb C^{d_2}\ot\cdots\ot\mathbb C^{d_n}$ are in general position.
If $k\le \sum_{i=1}^n (d_i-1)+1$ then they are linearly independent.
\end{proposition}

\begin{proof}
We use mathematical induction on $n$. By Proposition 2.1 in \cite{ha-kye-sep-face},
we know that the claim is true for $n=2$. Now, we show that if the claim holds for $n=N$,
then it also holds for $n=N+1$.  Let $1\le k\le \sum_{i=1}^{N+1} (d_i-1)+1$.
Suppose that $k$ product vectors
$$
|x_{1i},\,\ldots,\,x_{Ni},\,y_i\rangle,\quad (1\le i\le k)
$$
in $\mathbb C^{d_1}\ot\cdots\ot\mathbb C^{d_N}\ot\mathbb C^{d_{N+1}}$
are in general position, and have the relation
$$
\sum_{i=1}^k c_i |x_{1i},\, \ldots,\, x_{Ni},\, y_i\rangle =0,
$$
for scalars $c_1,\dots, c_k$.  So, we may assume that $k>d_{N+1}$. For any vector $|\omega\rangle$ in the orthogonal complement of the space
$\text{ \rm span}\{|y_i\rangle: 1\leq i \le d_{N+1}-1\}\subset \mathbb C^{d_{N+1}}$, we see that
\begin{equation}\label{cond1}
\langle \omega |y_i\rangle =0 \ \ (1\le i \le d_{N+1}-1),\qquad
\langle \omega |y_i\rangle \neq 0 \ \ (d_{N+1}\le i \le k),
\end{equation}
since product vectors are in general position. Then we have
\[
\sum_{i=d_{N+1}}^k c_i \langle w |y_i\rangle |x_{1i},\,\ldots,\,x_{Ni}\rangle =0.
\]
These $k-d_{N+1}+1$ product vectors $|x_{1i},\, \ldots, \, x_{Ni}\rangle$ $(d_{N+1}\le i \le k)$ are
linearly independent by the induction hypothesis, because $k-d_{N+1}+1\le \sum_{i=1}^N (d_i-1)+1$. Therefore, we see that
$c_i\langle \omega|y_i\rangle =0$, and so $c_i=0$ for each $i=d_{N+1},\cdots, k$ by \eqref{cond1}.
This implies that
$\sum_{i=1}^{d_{N+1}-1} c_i |x_{1i},\,\ldots,\,x_{Ni},\,y_i\rangle =0$,
and thus we see that $c_i=0$ for $1\le i \le d_{N+1}-1$ since $d_{N+1}-1$ vectors
$\{ |y_i\rangle : 1\le i \le d_{N+1}-1\}$ are linearly independent in $\mathbb C^{d_{N+1}}$. Consequently,
we get $c_i=0$ whenever $1\le i \le k$. This completes the proof.
\end{proof}

The number $\sum_{i=1}^n (d_i-1)+1$ in Proposition \ref{prop1} is optimal in two qutrit case, since
generic five dimensional spaces have six product vectors. It is also optimal in qubit-qudit case
by the construction in \cite{moment}.
Actually, tensor products of $(1,t)^{\rm t}\in\mathbb C^2$ and $(1,t,\dots, t^{d-1})^{\rm t}\in\mathbb C^d$ with $t\in\mathbb R$
span $(d+1)$-dimensional space. In $n$-qubit case, we consider $n$-times tensor product of
$(1,t)^{\rm t}\in\mathbb C^2$ which span $(n+1)$-dimensional space, to see that this number is also optimal.

The following theorem tells us that if $k\le \sum_{i=1}^n (d_i-1) $ then $k$ product vectors in general position
give rise to a simplicial face. The number $\sum_{i=1}^n (d_i-1)$ is also optimal in the above cases by the same examples.

\begin{theorem}\label{gp_simplex}
Suppose that $k$ product vectors in $\mathbb C^{d_1}\ot \mathbb C^{d_2}\ot\cdots\ot\mathbb C^{d_n}$ are in general position.
If $k\le \sum_{i=1}^n (d_i-1) $ then the span of these product vectors
has no more product vectors except for scalar multiples of these product vectors.
\end{theorem}

\begin{proof}
We also use induction on $n$. The assertion is true for $n=2$ by
Lemma 29 in \cite{chen_dj_ext_PPT}. We show that if the assertion is
true for $n=N$, then it is also true for $n=N+1$. We denote by $\mathcal
V$ the subspace of
$\mathbb C^{d_1}\otimes \cdots \otimes \mathbb C^{d_N}\otimes \mathbb C^{d_{N+1}}$
generated by $k$ product vectors $|x_{1i},\, \ldots,\,x_{Ni},\,y_{i}\rangle$ in general position with $k \le
\sum_{i=1}^{N+1}(d_i -1)$. Take a product
vector $|a_1,\,\ldots,\,a_N,\, b\rangle$ in the space $\mathcal V$, and write
\begin{equation}\label{pv}
|a_1,\,\ldots,\,a_N,\,b\rangle =
\sum_{i=1}^K c_i |x_{1i},\,\ldots,\,x_{Ni},\,y_i\rangle
\end{equation}
for scalars $c_i\in \mathbb C$ with $K\le k$, by
rearrangement of product vectors. We proceed to show that
$|a_1,\, \ldots,\, a_N,\,b\rangle
\in \mathcal V$ is a scalar multiple of $|x_{1i},\,\ldots,\,x_{Ni},\, y_i\rangle$, for some $i$
with $1\le i\le K$. The first step is to show the following:
\begin{equation}\label{sss}
\begin{aligned}
&|a_1,\,\ldots,\,a_N\rangle\ \text{\rm is a scalar
multiplication of a product vector}\\
&\phantom{XXXX}|x_{1i},\,\ldots,\,x_{Ni}\rangle\
\text{\rm for some}\ i=1,2,\dots,K.
\end{aligned}
\end{equation}

To do this, we first consider the case $K\le d_{N+1}$. In this case,
$|y_1\rangle, \dots, |y_K\rangle$ are linearly independent since
$|x_{1i},\,\ldots,\,x_{Ni},\,y_i\rangle$'s are in general position. We note that
$$
\left(\prod_{j=1}^N \langle a_j|a_j\rangle \right)|b\rangle =
\sum_{i=1}^K \left (c_i \prod_{j=1}^N \langle a_j|x_{ji}\rangle
\right) |y_i\rangle,
$$
and so, we see that $|b\rangle =\sum_{i=1}^K \xi_i |y_i\rangle$ with $\xi_i\in \mathbb C$. By (\ref{pv}), we have
\[
\sum_{i=1}^K \Bigl(\xi_i |a_1,\,\ldots,\,a_N\rangle -c_i |x_{1i},\,\ldots,\,x_{Ni}\rangle \Bigr)\otimes |y_i\rangle =0.
\]
From the linear independence of $|y_1\rangle, \dots, |y_K\rangle$,
we see that
$$
\xi_i |a_1,\,\ldots,\,a_N\rangle -c_i |x_{1i},\,\ldots,\,x_{Ni}\rangle=0
$$
for each $i=1,\dots, K$. Because not all of $\xi_i$'s are zero, the
assertion (\ref{sss}) follows.

Now, we consider the case $K> d_{N+1}$ to prove (\ref{sss}). Since $\{|y_1\rangle,\dots,
|y_{d_{N+1}}\rangle\}$ is a basis of $\mathbb C^{d_{N+1}}$, we have
$$
|y_i\rangle =\sum_{j=1}^{d_{N+1}} \eta_{ij}|y_j\rangle,\ \
(d_{N+1}+1\le i \le K), \qquad |b\rangle =\sum_{j=1}^{d_{N+1}} \zeta_j
|y_j\rangle,
$$
for scalars $\eta_{ij},\,\zeta_j\in\mathbb C$.
So, the relation (\ref{pv}) is written by
$$
\begin{aligned}
&\sum_{j=1}^{d_{N+1}} \left ( \zeta_j |a_1,\,\ldots,\,a_N\rangle\right)\otimes|y_j\rangle \\
=&\sum_{j=1}^{d_{N+1}} \biggl( c_j |x_{1j},\,\ldots,\,x_{Nj}\rangle
+\sum_{i=d_{N+1}+1}^K c_i \eta_{ij}|x_{1i},\,\ldots,\,x_{Ni}\rangle \biggr )\otimes |y_j\rangle.
\end{aligned}
$$
We take $\alpha\in\{1,2,\dots, d_{N+1}\}$ with
$\zeta_\alpha\neq 0$. By the linear independence of $|y_1\rangle,
\dots, |y_{d_{N+1}}\rangle$ again, we conclude that $|a_1,\,\ldots,\,a_N\rangle$ is contained in the  subspace $\mathcal
W\subset \mathbb C^{d_1}\otimes \cdots \ot \mathbb C^{d_N}$
generated by $K-d_{N+1}+1$ product vectors
$$
|x_{1\alpha},\,\ldots,\,x_{N\alpha}\rangle,\qquad
|x_{1i},\,\ldots,\,x_{Ni}\rangle \ \
(d_{N+1}+1\le i \le K).
$$
Since $K-d_{N+1}+1\le k-d_{N+1}+1\le \sum_{i=1}^N (d_i-1)$, the
induction hypothesis tells us that the product vector $|a_1,\,\ldots,\,a_N\rangle$ is a scalar multiple of
$|x_{1\beta},\,\ldots,\,x_{N\beta}\rangle $ for
some $\beta\in \{\alpha\}\cup \{ d_{N+1}+1,\cdots, K\}$, and this
completes the proof of statement (\ref{sss}).

Now, we may assume $\beta=1$, that is
\begin{equation}\label{pv2}
|a_1,\,\ldots,\,a_N,\,b\rangle
= c|x_{11},\,\ldots,\,x_{N1},\,b\rangle,
\end{equation}
with a nonzero scalar $c\in\mathbb C$,
by rearranging $K$ product vectors $|x_{1i},\,\ldots,\,x_{Ni},\,y_i\rangle $ with $1\le i\le K$.
Therefore, we have
\begin{equation}\label{pv3}
c|x_{11}\rangle \otimes |x_{21},\,\ldots,\,x_{N1},\,b\rangle \\
=\sum_{i=1}^{K} |x_{1i}\rangle \otimes\Bigl(  c_i|x_{2i},\,\ldots,\,x_{Ni},\,y_i\rangle  \Bigr),
\end{equation}
by (\ref{pv}) and (\ref{pv2}).

If $K\le d_1$ then we see that $|x_{11}\rangle, \dots, |x_{1K}\rangle$ are linearly independent, and so
we have $c_i=0$ for all $i=2,\cdots,k$ in (\ref{pv}).
Therefore, it remains to consider the case $K>d_1$. In this case,
we note that $\{ |x_{11}\rangle, \dots, |x_{1d_1}\rangle \}$
is a basis of $\mathbb C^{d_1}$. So, we have
\[
|x_{1j}\rangle = \sum_{i=1}^{d_1}\mu_{ij}|x_{1i}\rangle\quad
(d_1+1\le j \le K)
\]
with scalars $\mu_{ij}\in\mathbb C$.
We note that all $\mu_{ij}$'s are nonzero since the product vectors
$\{|x_{1i}\rangle:i=1,2,\dots, K\}$ are in general position.
Thus we have
\begin{equation*}
\begin{aligned}
&c|x_{11}\rangle \otimes |x_{21},\,\ldots,\,x_{N1},\,b\rangle \\
=&\sum_{i=1}^{d_1} |x_{1i}\rangle \otimes\biggl(  c_i|x_{2i},\,\ldots,\,x_{Ni},\,y_i\rangle +
\sum_{j=d_1+1}^K c_j \mu_{ij} |x_{2j},\,\ldots,\,x_{Nj},\,y_j\rangle \biggr),
\end{aligned}
\end{equation*}
by (\ref{pv3}).
From the linear independence of $|x_{11}\rangle, \dots, |x_{1d_1}\rangle$, we have
\begin{equation}\label{comp2}
c_i|x_{2i},\,\ldots,\,x_{Ni},\,y_i\rangle
+\sum_{j=d_1+1}^K c_j \mu_{ij}|x_{2j},\,\ldots,\,x_{Nj},\,y_j\rangle=0,
\end{equation}
for $\ i=2,\cdots, d_1$. Since $K-d_1+1\le k-(d_1-1)\le
\sum_{i=2}^{N+1}(d_i-1)$, we see that  $K-d_1+1$ product vectors
in \eqref{comp2} are linearly independent in $\mathbb C^{d_2}\otimes
\cdots \otimes\mathbb C^{d_{N+1}}$ by Proposition \ref{prop1}.
Recall that all $\mu_{ij}$ are nonzero. Therefore, we conclude that
$c_i=0$ for all $i=2,\cdots,k$ in (\ref{pv}), and this completes the
proof.
\end{proof}

If three product vectors in $\mathbb C^2\ot\mathbb C^2$ are in general position then the span of them has always
infinitely many product vectors. See \cite{moment}. The span of four product vectors in $\mathbb C^2\ot\mathbb C^2\otimes \mathbb C^2$
may also have infinitely many product vectors.
We denote by $|z_t\rangle$ the product vector in $\mathbb C^2\ot\mathbb C^2\ot\mathbb C^2$
given by three times repeated tensor product of $(1,t)^{\rm t}\in\mathbb C^2$. Four such product vectors are in general position, and
their span has all $|z_t\rangle$'s. But, such examples are very rare in the three qubit case. Actually, generic choice of four product vectors
$\{|x_{1i}\rangle \ot |x_{2i}\rangle \ot |x_{3i}\rangle\}$ in general position gives us linearly independent vectors
$\{|x_{2i}\rangle\ot |x_{3i}\rangle\}$ in $\mathbb C^4$. In this case,
we may apply Theorem \ref{gp_simplex} for $\mathbb C^2\ot\mathbb C^4$ to conclude that there are no more product
vectors in their span. We also see that the set $\{|z_t\rangle\langle z_t|:t\in\mathbb R\}$ of product states spans the seven dimensional
subspace in the real vector space of all $8\times 8$ self-adjoint matrices.
This shows that the number $\sum_{i=1}^n 2(d_i-1)+1$ is the best choice in the following proposition.

\begin{proposition}
Suppose that $k$ product vectors in $\mathbb C^{d_1}\ot \mathbb C^{d_2}\ot\cdots\ot\mathbb C^{d_n}$
are in general position.
If $k\le \sum_{i=1}^n 2(d_i-1)+1$ then the corresponding product states are linearly independent.
\end{proposition}
\begin{proof}Since pure product states corresponding to linearly independent product vectors are always
linearly independent, this assertion is true for $1\le k \le \sum_{i=1}^n (d_i-1)+1$ by Proposition \ref{prop1}.

Now, we consider the case of $\sum_{i=1}^n (d_i-1)+1< k\le \sum_{i=1}^n 2(d_i-1)+1$.
Let $\mathcal V_i\subset \mathbb C^{d_i}$ be the subspace generated by $d_i-1$ product vectors as follows:
\begin{equation*}
\begin{aligned}
\mathcal V_1&=\text{\rm span}\{|x_{1i}\rangle : 1\le i\le d_1-1\},\\
\mathcal V_j&=\text{\rm span}\left\{|x_{ji}\rangle: \sum_{\ell=1}^{j-1}
(d_{\ell}-1)+1 \le i \le \sum_{\ell=1}^{j} (d_{\ell}-1)\right\},\qquad j=2,3,\cdots, n.
\end{aligned}
\end{equation*}
We choose $n$ vectors
$|\psi_j\rangle \in \mathbb C^{d_j}$ such that $|\psi_j\rangle \in \mathcal V_j^{\perp}$ for each $j=1,2,\cdots,n$. We note that
\[
\sum_{\ell=1}^n (d_{\ell}-1)+1\le i\le k\, \Longrightarrow\, \prod_{j=1}^n \langle x_{j i}|\psi_{j}\rangle\neq 0,
\]
since the given $k$ product vectors are in general position.

Now, we put $|\psi\rangle =|\psi_1,\,\psi_2,\,\ldots,\,\psi_n\rangle$ and
$|z_i\rangle =|x_{1i},\,x_{2i},\,\ldots,\,x_{ni}\rangle$ for each $i=1,2,\cdots, k$. We suppose that
$\sum_{i=1}^k c_i|z_i\rangle \langle z_i|=0$. Then we have that
\[
0=\sum_{i=1}^k c_i|z_i\rangle \langle z_i|\psi\rangle
=\sum_{i=\sum_{\ell=1}^n (d_{\ell}-1)+1}^k \left(c_i \prod_{j=1}^n \langle x_{ji}|\psi_j\rangle \right)
|x_{1i},\,\ldots,x_{ni}\rangle.
\]
Since $k-\sum_{\ell=1}^n (d_{\ell}-1) \le \sum_{\ell=1}^n (d_{\ell}-1)+1$,
the product vectors in the righthand side of the above equality are linearly
independent by Proposition \ref{prop1}. Therefore, we have $c_i=0$ whenever
$\sum_{\ell=1}^n (d_{\ell}-1)+1\le i \le k$. Since $\sum_{\ell=1}^n (d_{\ell}-1)$ product states
$|z_i\rangle \langle z_i|$ for $1\le i \le \sum_{\ell=1}^n (d_{\ell}-1)$ are linearly independent
as stated in the beginning of this proof, we see that $c_i=0$ for all $1\le i \le \sum_{\ell=1}^n (d_{\ell}-1)$.
This completes the proof.
\end{proof}

\section{Three qubit case}\label{sec-3-qubit}

In the three qubit case, three product vectors in general position make a simplicial face by Theorem \ref{gp_simplex}.
The following theorem gives us a necessary and sufficient condition on four product vectors in general position
for which they make a simplicial face.

\begin{theorem}\label{theorem_gp}
Let $\{|z_i\rangle=|x_{1i}\rangle\ot |x_{2i}\rangle\ot |x_{3i}\rangle:i=1,2,3,4\}$
be in general position in $\mathbb C^2\ot\mathbb C^2\ot\mathbb C^2$.
Then they make a simplicial face if and only if $\{|x_{ji}\rangle\otimes |x_{ki}\rangle:i=1,2,3,4\}$
is linearly independent in $\mathbb C^2\ot\mathbb C^2$
for some $(j,k)=(1,2),\,(2,3),\,(3,1)$.
\end{theorem}

This theorem follows from the following proposition, because four product vectors
in general position are linearly independent by Proposition \ref{prop1} and
$\{|x_{ji}\rangle\otimes |x_{ki}\rangle:i=1,2,3,4\}$ spans at least three dimensional space by Proposition \ref{prop1} again.

\begin{proposition}\label{lemma_gp}
Let $\{|z_i\rangle=|x_{1i}\rangle\ot |x_{2i}\rangle\ot |x_{3i}\rangle:i=1,2,3,4\}$ be
in general position in $\mathbb C^2\ot\mathbb C^2\ot\mathbb C^2$,
and $D$ the subspace of $\mathbb C^2\ot\mathbb C^2\ot\mathbb C^2$ spanned by these product vectors.
Then the following are equivalent:
\begin{itemize}
\item[(i)] $D$ has a product vector which is not parallel to any $|z_i\rangle$.
\item[(ii)] $\text{\rm dim(span}\{|x_{ji}\rangle\otimes |x_{ki}\rangle:i=1,2,3,4\})=3$ for every $(j,k)=(1,2),\,(2,3),\,(3,1)$.
\item[(iii)] $D$ has infinitely many product vectors.
\end{itemize}
\end{proposition}

\begin{proof}
(i) $\Longrightarrow$ (ii): We first show that $\{|x_{2i},\,x_{3i}\rangle:i=1,2,3,4\}$ is linearly dependent.
To see this, we write
$|\alpha,\,\beta,\, \gamma\rangle=\sum_{i=1}^4 a_i |x_{1i},\,x_{2i},\,x_{3i}\rangle$, where at least two of $a_i$'s are nonzero.
Take a nonzero vector $|\alpha^\perp\rangle$
which is orthogonal to $|\alpha\rangle$, then we have
$$
\sum_{i=1}^4 a_i\langle \alpha^\perp|x_{1i}\rangle |x_{2i},\,x_{3i}\rangle=0.
$$
We note that $\langle \alpha^\perp|x_{1i}\rangle\neq 0$ for at least three $i=1,2,3,4$.
Therefore, we see  that $\{|x_{2i},\,x_{3i}\rangle:i=1,2,3,4\}$ are linearly dependent, and have the required conclusion
by Proposition \ref{prop1} for $\mathbb C^2\ot\mathbb C^2$.

(ii) $\Longrightarrow$ (iii): We identify $\mathbb C^2\ot\mathbb C^2\ot\mathbb C^2$ with $\mathbb C^4\ot\mathbb C^2$
in the obvious way.
Take  $|\zeta\rangle\in\mathbb C^4$ which is orthogonal to the span of $\{|x_{1i},\,x_{2i}\rangle:i=1,2,3,4\}$.
Then $|\zeta\rangle\ot |e_1\rangle$ and $|\zeta\rangle\ot |e_2\rangle$ are orthogonal to $D$.
By the same reasoning, we can also take a non-product vector $|\xi\rangle\in\mathbb C^4$ such that
$|e_1\rangle\ot |\xi\rangle$ and $|e_2\rangle\ot |\xi\rangle$
are orthogonal to $D$. It is easy to see that $\{|\zeta\rangle\ot |e_i\rangle,\, |e_i\rangle\ot |\xi\rangle: i=1,2\}$
is a basis of $D^\perp$. Indeed, if the intersection of two subspaces
$\spa\{|\zeta\rangle\otimes |e_1\rangle,\,|\zeta\rangle \otimes |e_2\rangle\}$
and $\spa\{|e_1\rangle \otimes |\xi\rangle,\, |e_2\rangle\otimes |\xi\rangle\}$
has a nonzero vector then it must be a product vector in $\mathbb C^2\otimes \mathbb C^2\otimes \mathbb C^2$.
In such a case, the nonzero vector $|\zeta\rangle$ should be a product vector in $\mathbb C^2\ot\mathbb C^2$.
But, this is impossible since
$\{|x_{1i}\rangle\otimes |x_{2i}\rangle:i=1,2,3,4\}$ is in general position in $\mathbb C^2\otimes \mathbb C^2$.

Now, for each $|\beta\rangle\in \mathbb C^2$, there exist $|\alpha\rangle$ and $|\gamma\rangle$ such that
$ |\alpha\rangle \ot |\beta\rangle$ is orthogonal to $|\zeta\rangle$
and $|\beta\rangle\ot |\gamma\rangle$ is orthogonal to $|\xi\rangle$. Therefore, we see that there are infinitely many product vectors in $D$.

There is nothing to prove for the direction (iii) $\Longrightarrow$  (i).
\end{proof}

It was shown in \cite{bra} that any unextendible product basis in $\mathbb C^2\ot\mathbb C^2\ot\mathbb C^2$ has exactly four vectors,
and there exists no other product vectors in their span. Therefore,
they must be in general position as product vectors in $\mathbb C^2\ot\mathbb C^4$ by Proposition \ref{lemma_gp}.
This can be seen directly, as follows.

\begin{proposition}
Let four product vectors form an unextendible product basis in the
space $\mathbb C^2\ot \mathbb C^2\ot\mathbb C^2$. Then they are in
general position as vectors in $\mathbb C^2\ot \mathbb C^4$, as well
as in $\mathbb C^2\ot\mathbb C^2\ot\mathbb C^2$.
\end{proposition}

\begin{proof}
It is easy to see that four product vectors are of the form
$$
|x\rangle\ot |\eta\rangle\ot |z\rangle,\quad
|x^\perp\rangle \ot |y\rangle \ot |\zeta\rangle,\quad
|\xi\rangle \ot |y^\perp\rangle \ot |z^\perp\rangle,\quad
|\xi^\perp\rangle \ot |\eta^\perp\rangle \ot |\zeta^\perp\rangle
$$
for $|x\rangle,\,|y\rangle,\,|z\rangle,\,|\xi\rangle,\,|\eta\rangle,\,|\zeta\rangle\in\mathbb C^2$,
where $|a^\perp\rangle$ denotes the vector orthogonal to $|a\rangle$ in $\mathbb C^2$. To see that
the following four vectors
$$
|\eta,\,z\rangle, \qquad
|y,\, \zeta\rangle,\qquad
|y^\perp,\, z^\perp\rangle,\qquad
|\eta^\perp,\, \zeta^\perp\rangle
$$
are linearly independent, we write $|\eta\rangle,\,|z\rangle,\,|\zeta\rangle$ as a linear combination of $|y\rangle$ and $|y^\perp\rangle$,
and consider the $4\times 4$ matrix whose rows are given by coefficients of the above four vectors.
By a direct calculation, we see that the determinant is nonzero.
\end{proof}

We proceed to consider five dimensional subspaces. Recall that generic five dimensional subspaces have
exactly six product vectors. We show that these six product vectors make a simplicial face.
Recall that every six dimensional subspace of $\mathbb C^2\ot\mathbb C^2\ot\mathbb C^2$ has infinitely many product vectors
by \cite{juhan}. We state our main results in this section:

\begin{theorem}\label{main-th}
Suppose that a five dimensional subspace $D$
of $\mathbb C^2\otimes \mathbb C^2\otimes \mathbb C^2$ has exactly six
product vectors. Then we have the followings:
\begin{enumerate}
\item[{\rm (i)}]
Any five product vectors among six product vectors are linearly independent.
\item[{\rm (ii)}]
The corresponding six pure product states are linearly independent.
\end{enumerate}
\end{theorem}

For the proof, we begin with two simple lemmas.

\begin{lemma}\label{lem_pdv}
Suppose that a subspace $D$ of $\mathbb C^{d_1}\otimes \mathbb C^{d_2}\otimes \mathbb C^{d_3}$ has
a finite number of product vectors $|z_i\rangle =|x_{1i},\,x_{2i},\,x_{3i}\rangle$
which are not parallel to each others.
For two product vectors $|z_i\rangle,\, |z_j\rangle $ in $D$, if $|x_{ki}\rangle$ is parallel to $|x_{kj}\rangle$
for some $k\in \{1,2,3\} $, then  $|x_{\ell i}\rangle$ is not parallel to $|x_{\ell j}\rangle $
for each $\ell \in \{1,2,3\}\setminus \{k\}$.
\end{lemma}

\begin{proof}
If not, we may assume that $|x_{ki}\rangle$ is parallel to $|x_{kj}\rangle$ for $k=1,2$. In this case,
 $|x_{3i}\rangle$ is not parallel to $|x_{3j}\rangle$, and so we see that $\beta |z_i\rangle
+\gamma |z_j\rangle =|x_{i1},\,x_{i2}\rangle \otimes
(\beta |x_{3i}\rangle+\gamma |x_{3j}\rangle)$ is a product vector in
$D$ for any complex numbers $\beta,\, \gamma$. This is a
contradiction.
\end{proof}

\begin{lemma}\label{lll}
Three product vectors $\{|x_i, y_i\rangle:i=1,2,3\}$ in $\mathbb C^2\ot\mathbb C^2$ must be linearly independent,
whenever either $\spa\{|x_1\rangle,|x_2\rangle,|x_3\rangle\}$ or $\spa\{|y_1\rangle,|y_2\rangle,|y_3\rangle\}$ is of two dimensional.
\end{lemma}
\begin{proof}
If they are in general position, then this is Proposition \ref{prop1}. Suppose that two,
say $|x_1\rangle$ and $|x_2\rangle$, are parallel to each others.
If $\sum_{i=1}^3 a_i|x_i,\,y_i\rangle=0$ then take nonzero vector $|x\rangle$ so that $\langle x|x_i\rangle=0$ for $i=1,2$.
Now, we consider $\sum_{i=1}^3 a_i\langle x|x_i\rangle |y_i\rangle =0$ to get
$a_3=0$.
\end{proof}

The following proposition proves the statement (i) of Theorem~\ref{main-th} in the case
when six product vectors form a generalized unextendible product basis, that is,
the orthogonal complement $D^\perp$ of $D$ has no product vectors.

\begin{proposition}\label{6-theorem}
Suppose that a $5$-dimensional subspace $D$ of $\mathbb
C^2\ot\mathbb C^2\ot\mathbb C^2$ has exactly six product vectors $\{|z_i\rangle=|x_{1i},\,x_{2i},\, x_{3i}\rangle:i=1,2,\dots,6\}$
which form a generalized unextendible product basis. Then
any five product vectors among six product vectors are linearly independent.
\end{proposition}

\begin{proof}
We first show that any five of them also form a generalized unextendible product basis.
By Proposition \ref{gupb}, it is enough to prove the following:
\begin{enumerate}
\item[{\rm (i)}]
If $S$ is a subset of $\{1,2,\dots,6\}$ with $|S|=3$ then $\{|x_{ji}\rangle: i\in S\}$ spans $\mathbb C^2$ for each $j=1,2,3$.
\item[{\rm (ii)}]
If $S_1$ and $S_2$ are disjoint subsets of $\{1,2,\dots,6\}$ with $|S_1|=|S_2|=2$ then $\{|x_{ji}\rangle: i\in S_j\}$ spans $\mathbb C^2$
for some $j=1,2$.
\end{enumerate}
For the proof of (i),
we first note that no four vectors among $\{|x_{1i}\rangle:1,2,\dots,6\}$ can be parallel to each others by Proposition \ref{gupb}.
Suppose that three vectors, say $\{|x_{11}\rangle,\,|x_{12}\rangle,\, |x_{13}\rangle\}$, are parallel to each others.
In this case, $\{|x_{2i},\, x_{3i}\rangle:i=4,5,6\}$ in $\mathbb C^2\ot\mathbb C^2$ satisfies the assumption of Lemma \ref{lll}
by Proposition \ref{gupb} again. We also note that $\{|x_{2i},\,x_{3i}\rangle:i=1,2,3\}$ is in general position in
$\mathbb C^2\otimes \mathbb C^2$ by Lemma~\ref{lem_pdv}. Consequently, both $\{|z_i\rangle:i=1,2,3\}$ and $\{|z_i\rangle:i=4,5,6\}$ span
$3$-dimensional subspaces of $\mathbb C^2\otimes \mathbb C^2 \otimes \mathbb C^2$, respectively.
We denote by $D_1$ and $D_2$ the spans of
$\{|z_j\rangle: j=1,2,3\}$ and $\{|z_j\rangle:j=4,5,6\}$, respectively. If there exists a vector in $D_1\cap D_2$, then we have
$$
\sum_{j=1}^3 a_i |x_{1i},\,x_{2i},\,x_{3i}\rangle = \sum_{j=4}^6 a_i |x_{1i},\,x_{2i},\,x_{3i}\rangle.
$$
We take $|x_1\rangle$ orthogonal to $|x_{1i}\rangle$ for $i=1,2,3$, to get
$$
\sum_{j=4}^6 a_i \langle x_1|x_{1i}\rangle |x_{2i}\ot x_{3i}\rangle=0.
$$
Since $\{|x_{2i},\,x_{3i}\rangle:i=4,5,6\}$ in $\mathbb C^2\ot\mathbb C^2$ is linearly independent by Lemma \ref{lll}
and $\langle x_1|x_{1i}\rangle\neq0$, we have
$a_4=a_5=a_6=0$. This tells us that two space $D_1$ and $D_2$ have no nonzero intersection, and $D$ is of six dimensional.

To prove (ii), it suffices to consider the case when any three vectors among
$\{|x_{ji}\rangle:i=1,2,\dots,6\}$ span $\mathbb C^2$  for each $j=1,2,3$  by the above result (i).
Assume that there exist disjoint subsets, say $\{1,2\}$ and $\{3,4\}$ such that both $\{|x_{11}\rangle,\,|x_{12}\rangle\}$
and $\{|x_{23}\rangle,\,|x_{24}\rangle\}$ span
one dimensional spaces. In this case, we see that $\{|x_{35}\rangle,\,|x_{36}\rangle\}$ is linearly independent in
$\mathbb C^2$ by Proposition \ref{gupb}. Then, it is easy to see that the vectors $|z_1\rangle,\,|z_2\rangle,\,|z_3\rangle,\,|z_4\rangle$
are linearly independent by Lemma \ref{lem_pdv}.
Suppose that
$$
\sum_{i=1}^4 a_i |x_{1i},\, x_{2i},\, x_{3i}\rangle=\sum_{i=5}^6 a_i |x_{1i},\, x_{2i},\, x_{3i}\rangle.
$$
We take two vectors $|x_1\rangle$ and $|x_2\rangle$, which span the orthogonal complements of $|x_{11}\rangle$ and
$|x_{23}\rangle$, respectively.
Applying $\langle x_1,\, x_2|$ to both sides of the above equation,
we have
$\sum_{i=5}^6 a_i \langle x_1|x_{1i}\rangle \langle x_2|x_{2i}\rangle |x_{3i}\rangle=0$.
Since $\langle x_1|x_{1i}\rangle \langle x_2|x_{2i}\rangle \neq 0$ for $i=5,6$ by the above result (i), we see that $a_5=a_6=0$.
This shows that the span of six product vectors is of six dimensional.
Therefore, we have shown that any five product vectors also form
a generalized unextendible product basis.

We proceed to show that they are linearly independent.
Suppose that they are linearly dependent, and so they span a four dimensional subspace $E$. Take four vectors which span $E$. Then
$E$ has another extra product vector, but does not have infinitely many product vectors.
This shows that these four product vectors are not in general position
by Proposition \ref{lemma_gp}, and $E^\perp$ has a product vector by Proposition \ref{gp}.
This tells us that the five vectors do not form a generalized unextendible product basis.
\end{proof}

Now, we consider the case when $E:=\{|z_i\rangle=|x_{1i},\,
x_{2i},\, x_{3i}\rangle:i=1,2,\ldots,6\}\subset D$ does not form a
generalized unextendible product basis, that is, $E^{\perp}$ has a
product vector. To begin with, we recall \cite{moment} that if a
subspace of $\mathbb C^2\ot\mathbb C^2$ has three product vectors
then it must have infinitely many product vectors. If $E^\perp$ has
a product vector then there is a partition $I_1\cup I_2\cup I_3$ of
$\{1,2,\dots,6\}$ so that $\{|x_{ji}\rangle:i\in I_j\}$ spans one
dimensional space for each $j=1,2,3$ by Proposition \ref{gupb}.
Suppose that $|I_j|\ge 3$ for some $j=1,2,3$. Then by the result
\cite{moment} mentioned above, we see that $D$ has infinitely
product vectors. Therefore, we conclude that if $D$ satisfies the
condition of Theorem \ref{main-th} then $|I_j|=2$ for each
$j=1,2,3$. Therefore, the following proposition completes the proof of the
statement (i) of Theorem \ref{main-th}.

\begin{proposition}\label{non-gupb}
Suppose that six product vectors $\{|z_i\rangle=|x_{1i},\,x_{2i},\, x_{3i}\rangle:i=1,2,\dots,6\}$
are given with the following properties:
\begin{itemize}
\item[{\rm (i)}] $|x_{1i}\rangle$ is parallel to $|x_{11}\rangle$ if and only if $i=2$.
\item[{\rm (ii)}] $|x_{2i}\rangle$ is parallel to $|x_{23}\rangle$ if and only if $i=4$.
\item[{\rm (iii)}] $|x_{3i}\rangle$ is parallel to $|x_{35}\rangle$ if and only if $i=6$.
\end{itemize}
If the span of these six product vectors have finitely many product vectors, then any five of six product vectors are linearly independent.
\end{proposition}

\begin{proof}
It suffices to consider the  five product vectors $|z_i\rangle$ with $i=1,2,\ldots,5$.
Suppose that $\sum_{i=1}^5 a_i |x_{1i},\,x_{2i},\, x_{3i}\rangle=0$.
We take $|x_1\rangle$ and $|x_2\rangle$ orthogonal to $|x_{11}\rangle$ and $|x_{23}\rangle$, respectively.
Then we have
$$
0=\sum_{i=1}^5 a_i \langle x_1|x_{1i}\rangle \langle x_2|x_{2i}\rangle |x_{3i}\rangle=
a_5 \langle x_1|x_{15}\rangle \langle x_2|x_{25}\rangle |x_{35}\rangle,
$$
from which we have $a_5=0$, and $\sum_{i=1}^4 a_i |x_{1i},\,x_{2i},\, x_{3i}\rangle=0$. Then we have
$$
a_1 \langle x_{11}|x_{11}\rangle \langle x_2|x_{21}\rangle |x_{31}\rangle
+a_2 \langle x_{11}|x_{11}\rangle \langle x_2|x_{22}\rangle |x_{32}\rangle=0.
$$
We note that there are infinitely many product vectors if $|x_{31}\rangle$ and $|x_{32}\rangle$ are linearly dependent,
and so $a_1=a_2=0$. From this we also have $a_3=a_4=0$.
\end{proof}

The following proposition now completes the proof of Theorem \ref{main-th}.
We include here Lin Chen's argument \cite{chen-pri} which is much simpler than our original proof.
We are grateful to him for informing this and allowing us to put it here.

\begin{proposition}
Suppose that five of six product vectors in $\mathbb C^2\ot\mathbb C^2\ot\mathbb C^2$
are linearly independent then the corresponding six pure product states are linearly independent.
\end{proposition}

\begin{proof}
We may assume that five vectors $|z_1\rangle,\ldots,|z_5\rangle$ are linealy independent. Suppose that
$|z_1\rangle \langle z_1|,\ldots,|z_6\rangle \langle z_6|$ are linearly dependent, then we have
$$
|z_6\rangle \langle z_6|=\sum_{i=1}^5 a_i |z_i\rangle \langle z_i|.
$$
The left hand side of the above equation has rank one, while the right hand side
has rank equal to the number of nonvanishing $a_i$ because $\{|z_i\rangle : i=1,2,\ldots,5\}$ is linearly independent.
So the only possibility is that exactly one $a_i$ is nonzero. It implies $|z_6\rangle$ and $|z_i\rangle$ are parallel,
which gives us a contradiction.
\end{proof}

Therefore, we get a simplicial face which is affinely isomorphic to the five dimensional simplex $\Delta_5$
with six vertices. A maximal face of $\Delta_5$ has five vertices whose corresponding product vectors
are linearly independent. Therefore, we can follow the argument in \cite{ha-kye-sep-face} to get a
PPT edge state of rank four as it will be explained in the next section.

Six product vector satisfying the condition of Proposition \ref{non-gupb} may span six or five dimensional spaces.
As for the case when they span the five dimensional space, we consider the following example, which shows that
six product vectors may make a simplicial face even though they do not form a generalized unextendible product basis.

\begin{example}\label{vec_ex}
Consider the following six product vectors given by
$$
\begin{aligned}
|z_1\rangle &=|e_2\rangle \otimes \bigl( |e_1\rangle+ 2|e_2\rangle\bigr)\otimes |e_1\rangle, \\
|z_2\rangle &=|e_2\rangle \otimes \bigl( |e_1\rangle+|e_2\rangle \bigr)\otimes \bigl(|e_1\rangle +|e_2\rangle \bigr),\\
|z_3\rangle &=|e_1\rangle \otimes |e_2\rangle \otimes \bigl( |e_1\rangle -|e_2\rangle\bigr),\\
|z_4\rangle &=\bigl(|e_1\rangle +|e_2\rangle \bigr)\otimes |e_2\rangle \otimes \bigl(|e_1\rangle +|e_2\rangle \bigr),\\
|z_5\rangle &=\bigl( |e_1\rangle +2|e_2\rangle\bigr)\otimes |e_1\rangle \otimes |e_2\rangle,\\
|z_6\rangle &=\bigl( |e_1\rangle +|e_2\rangle \bigr)\otimes \bigl(|e_1\rangle -2|e_2\rangle\bigr)\otimes |e_2\rangle.
\end{aligned}
$$
Then we have
$|z_1\rangle+|z_3\rangle+|z_5\rangle=|z_2\rangle+|z_4\rangle+|z_6\rangle$,
and so they span the five dimensional space $D$. These six product
vectors do not form a generalized unextendible product basis, since
$D^\perp$ has a product vector $|e_1\rangle\ot |e_1\rangle\ot
|e_1\rangle$. In fact, this is the only product vector in
$D^{\perp}$.  It is easy to check that $D$ has only these six
product vectors up to scalar multiplications.
\end{example}

\section{Construction of three qubit PPT entanglement of rank four}

In this section, we construct three qubit PPT entangled states of rank four.
We begin with the description of facial structures of the convex set $\mathbb T$ consisting of all PPT states
acting on the Hilbert space ${\mathcal H}=\mathbb C^{d_1}\ot \mathbb C^{d_2}\ot\cdots\ot\mathbb C^{d_n}$.
For a given subset $S$ of $\{1,2,\dots,n\}$, we can define the linear map  $T(S)$
from $\bigotimes_{j=1}^n M_{d_j}$ into itself by
$$
(A_1\ot A_2\ot\cdots\ot A_n)^{T(S)}:=B_1\ot B_2\ot\cdots\ot B_n,
\quad \text{\rm with}\ B_j=\begin{cases} A_j^\ttt, &j\in S,\\ A_j,
&j\notin S,\end{cases}
$$
where $A^\ttt$ denotes the transpose of the matrix $A$. A state $\varrho$ is
said to be of positive partial transpose (PPT) if $\varrho^{T(S)}$ is positive for every subset
$S$. The PPT criterion \cite{choi-ppt,peres} tells us that a
separable state must be of PPT. The convex set $\mathbb
T$ is the intersection of convex sets
$$
\mathbb T^S=\{\varrho: \varrho^{T(S)} \ {\text{\rm is positive}}\}
$$ through subsets $S$ of
$\{1,2,\dots,n\}$. Since $\varrho^{T(S)}$ is positive if and only if
$\varrho^{T(S^c)}$ is positive with the complement $S^c$ of $S$ in
$\{1,2,\dots,n\}$, the convex set $\mathbb T$ is actually
intersection of $2^{n-1}$ convex sets $\mathbb T^S$, whose faces are
determined by subspaces of the Hilbert space $\mathcal H$. From
this, it is easy to describe their facial structures for the convex
set $\mathbb T$. See \cite{ha-kye-2} for the bi-partite case of
$n=2$. For a product vector $\ket z=\ket{x_1}\ot\cdots\ot\ket{x_n}$,
we also define the partial conjugate $\ket z^{\Gamma(S)}$ by
$$
(\ket{x_1}\ot\cdots\ot\ket{x_n})^{\Gamma(S)}
:=\ket{y_1}\ot\cdots\ot\ket{y_n}, \quad \text{\rm with}\
\ket{y_j}=\begin{cases} \ket{\bar{x}_j}, &j\in S,\\
\ket{x_j}, &j\notin S,\end{cases}
$$
where $|\bar x\rangle$ denotes the conjugate of $|x\rangle$.
We will abuse notations $T(i)$ and $\Gamma(i)$ for $T(\{i\})$ and $\Gamma(\{i\})$, respectively.

We will restrict ourselves to the three qubit case with $\mathcal H=\mathbb C^2\ot\mathbb C^2\ot\mathbb C^2$. For a
given quartet $(D_0,D_1,D_2,D_3)$ of subspaces of $\mathcal H$, the convex set
$$
\bigcap_{i=0}^3\{\varrho\in\mathbb T: {\mathcal R}\varrho^{T(i)}\subset D_i\}
$$
is a face of $\mathbb T$ unless it is empty, and every face of $\mathbb T$ is in this form,
where $\varrho^{T(0)}$ denotes $\varrho^{T(\emptyset)}=\varrho$ for notational convenience and ${\mathcal R}\varrho$ means the
range of $\varrho$. It is very difficult to determine if the above set is nonempty or not.
Many authors have been trying to classify PPT states by possible combinations of ranks of $\varrho^{T(i)}$.
See \cite{abls,bdmsst,gar,karnas}, for example. See also \cite{aug_n_qubit,johnston,tura} for higher qubit cases.
Most important cases are PPT entangled edge states, that is, PPT states with no product vectors in their ranges.
Such states in two qutrit case have been completely classified in \cite{kye-prod-vec,kye_osaka} by their ranks.
It is still open for $2\ot 4$ case. See \cite{kye_ritsu}. We refer to \cite{juhan}
for recent progress.

PPT states of rank four are of special interest since four is the lowest rank for PPTES. Structures of such
PPT states have been studied extensively in \cite{rank4}. Especially, if $\varrho$ is a three qubit PPT state of rank four then
$\varrho^{T(i)}$ is also of rank four for each $i=1,2,3$. Furthermore, it is separable if and only if the range has a
product vector. Therefore, any three qubit PPT entangled states of rank four must be edge states.

From now on, we suppose that six product vectors $\{|z_i\rangle:i=1,2,\dots,6\}$ in $\mathbb C^2\ot\mathbb C^2\ot\mathbb C^2$
span the $5$-dimensional subspace $D$.
We also assume that $\{|z_i\rangle^{\Gamma(j)}:i=1,2,\dots 6\}$ also spans a five dimensional space for each $j=1,2,3$.
This is the case if all of the entries are real numbers, as in the examples in Section \ref{gp_gupb} and Section \ref{sec-3-qubit}.
Therefore, we have a simplicial face $\Delta_5$ with six extreme points $|z_i\rangle\langle z_i|$.
We write
\begin{equation}\label{coeff}
|z_6\rangle=\sum_{i=1}^5 a_i|z_i\rangle
\end{equation}
with scalars $a_i$, and take positive numbers $p_i$ with $\sum_{i=1}^5p_i=1$. Consider
\begin{equation}\label{edge-4}
\varrho_t=(1-t)|z_6\rangle\langle z_6|+t\sum_{i=1}^5 p_i|z_i\rangle\langle z_i|,
\end{equation}
for real number $t$. We note that $\varrho_t$ is an interior point of the convex set $\Delta_5$ for $0<t<1$, and
$\varrho_1$ is an interior point of a maximal face $\Delta_4$ of $\Delta_5$. Because $\varrho_t$ is of rank five for $0<t\le 1$,
there exists $t>1$ such that $\varrho_t$ is still positive. This is the case for $\varrho_t^{T(j)}$ for each $j=1,2,3$.
Take the largest number $\lambda>1$ such that $\varrho_\lambda$ is of PPT. Then there exists
$j=0,1,2,3$ such that the rank of $\varrho_\lambda^{T(j)}$ is strictly less than five. For such $j=0,1,2,3$,
we see that $\varrho_\lambda^{T(j)}$ is entangled, since its range space has
only six product vectors $|w_i\rangle:=|z_i\rangle ^{\Gamma(j)}$ and the simplicial face
with six extreme points $|w_i\rangle \langle w_i|=\bigl (|z_i\rangle \langle z_i|\bigr )^{T(j)}$ does not contain $\varrho_\lambda^{T(j)}$.
By the results in \cite{rank4} mentioned just above, we see that $\varrho_\lambda^{T(j)}$ must be of rank four for each
$j=0,1,2,3$. Especially, we see that $\varrho_\lambda$ is of rank four. We proceed to determine this number $\lambda$
in terms of $p_i$ and $a_i$.

\begin{figure}[h!]
\begin{center}
\includegraphics[scale=0.7]{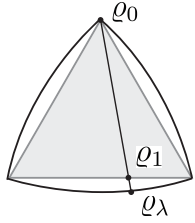}
\end{center}
\caption{The triangle, edges, vertices represent the $5$-simplex, $4$-simplices and $3$-simplices consisting of separable states, respectively.
The boundary points of the round convex body which are not on the vertices represent PPT entangled edge states of rank four. }
\end{figure}

Since the orthogonal complement $D^\perp$ is already in the kernel of
$\varrho_\lambda$ which has the four dimensional kernel, there exists a unique vector $\xi\in D\cap\ker\varrho_\lambda$
up to scalar multiplications. From the relation $\varrho_\lambda|\xi\rangle=0$, we
can express $|z_6\rangle$ as the linear combination of $|z_i\rangle$ with $i=1,2,\dots, 5$. We compare
these coefficients with (\ref{coeff}) to get the relation
$$
(1-t) a_i \langle z_6|\xi\rangle +tp_i \langle z_i|\xi\rangle =0,\qquad i=1,2,3,4,5.
$$
Because $\xi\in D$, we see that five vectors $\{(1-t)\bar a_i |z_6\rangle +tp_i |z_i\rangle: i=1,2,3,4,5\}$ are linearly dependent.
If we use the relation (\ref{coeff}) to express these vectors as linear combinations of $\{|z_i\rangle: i=1,2,3,4,5\}$ and
write $|a\rangle=(a_1,a_2,\ldots,a_5)^{\rm t}$, then
we have the coefficient matrix
$$
(1-\lambda)|a\rangle\langle a|+\lambda {\text{\rm Diag}}\, (p_1,p_2,p_3,p_4,p_5)
$$
which must be singular.
From the condition $\lambda>0$, we see that $\lambda$ satisfies the equation
$$
(1-\lambda)\left(\sum_{i=1}^5\dfrac{|a_i|^2}{p_i}\right) +\lambda=0
$$
with the null vector whose $i$-th entry is given by $\frac{a_i}{p_i}$.
 We note that the coefficients $a_i$'s are in fact already determined by
five product vectors $|z_i\rangle$ with $i=1,2,3,4,5$. Turning back to (\ref{edge-4}), we can express $\varrho_\lambda$ in terms of
$|z_i\rangle$ and $p_i$ for $i=1,2,3,4,5$.

For concrete examples, we consider
normalizations $|\tilde z_i\rangle$'s of six product vectors introduced in Section 2 after Proposition \ref{gp}.
For each $p=(p_1,p_2,p_3,p_4,p_5)$ with $\sum_{i=1}^5 p_i=1$ and $p_i>0$,
our construction gives us rank four PPT entangled states $\varrho_{p}$ by
\begin{equation}\label{pptes_ex}
\varrho_p=
\frac{1}{\alpha-1}\biggl(\alpha \sum_{i=1}^5 p_i|\tilde z_i\rangle \langle \tilde z_i|-
|\tilde z_6\rangle \langle \tilde z_6| \biggr),
\end{equation}
where $\alpha
=\frac{2}{9p_1}+\frac{2}{9p_2}+\frac{2}{9p_3}+\frac{8}{81p_4}+\frac{125}{81p_5}$.
We can describe four vectors spanning the kernel of $\varrho_p$ in terms of $p_i$'s, and
check the existence of product vectors in the kernel of $\varrho_p$.
See Appendix E in \cite{gar} for a general method of finding product vectors.
In contrast to PPT entangled states  constructed from unextendible product basis \cite{bdmsst},
the kernel of $\varrho_p$ contains no product vector in most cases of $p$ including
$p=(\frac 15,\frac 15,\frac 15,\frac 15,\frac 15)$.

For another examples, we take $|\tilde z_i\rangle$ as normalizations of six product vectors in Example~\ref{vec_ex},
to get get rank four PPT entangled states as in \eqref{pptes_ex} with
$\alpha=\frac 1{2p_1}+\frac 2{5p_2}+\frac 1{5p_3}+\frac 2{5p_4}+\frac 1{2p_5}$.
In this case, the kernel of $\varrho_p$ contains no product vector for any $p=(p_1,p_2,p_3,p_4,p_5)$ with $\sum_{i=1}^5 p_i=1$ and $p_i>0$.


In our construction, one crucial condition is that partial conjugates of six product vectors must
span five dimensional spaces. If one kind of partial conjugates of six product vectors are linearly independent, then
we cannot obtain PPT entangled states by the above method. For example, we consider the following six product vectors:
$$
\begin{aligned}
|w_1\rangle &=\bigl( |e_1\rangle+|e_2\rangle \bigr) \otimes \bigl( |e_1\rangle+|e_2\rangle \bigr) \otimes \bigl( |e_1\rangle+|e_2\rangle \bigr),\\
|w_2\rangle &=\bigl( |e_1\rangle+i|e_2\rangle \bigr) \otimes \bigl( |e_1\rangle-|e_2\rangle \bigr) \otimes \bigl( |e_1\rangle-i|e_2\rangle \bigr),\\
|w_3\rangle &=\bigl( |e_1\rangle-|e_2\rangle \bigr) \otimes \bigl( |e_1\rangle+|e_2\rangle \bigr) \otimes \bigl( |e_1\rangle-|e_2\rangle \bigr),\\
|w_4\rangle &=\bigl( |e_1\rangle-i|e_2\rangle \bigr) \otimes \bigl( |e_1\rangle-|e_2\rangle \bigr) \otimes \bigl( |e_1\rangle+i|e_2\rangle \bigr),\\
|w_5\rangle &=|e_1\rangle \otimes |e_1\rangle \otimes |e_1\rangle,\\
|w_6\rangle &=|e_2\rangle \otimes |e_1\rangle \otimes |e_2\rangle,
\end{aligned}
$$
which span the five dimensional subspace with completely entangled orthogonal complement.
We note that their partial conjugates span six dimensional space.


\section{higher qubit cases and discussion}

It is an interesting question to ask to what extent our approach
works. To get separable states with unique decomposition, we have
exploited the fact that generic $s_{\max} +1$ dimensional subspaces have finitely many product vectors
which give rise to linearly independent product states, where $s_{\max}$ is given by (\ref{s_max}).
In the $n$ qubit case, generic $2^n-n$ dimensional spaces have $n!$ product states in the $2^{2n}$ dimensional
real vector space
consisting of all $2^n\times 2^n$ self-adjoint matrices.
If $n\ge 9$ then $n!>2^{2n}$, and so the convex hull of these $n!$ pure product states is not a simplex anymore, but a
polytope with $n!$ extreme points. It would be interesting to investigate combinatorial
structures of these polytopes.

By Theorem \ref{gp_simplex}, we see that $n$ qubit separable states with rank $\le n$ have unique decomposition in most cases.
Our next question is what happens for $n$ qubit separable states of rank $k$ with $n<k<s_{\max}+1$.
We note that Theorem \ref{theorem_gp} gives an answer for three qubit case.

One important application of our results is to construct three qubit PPT entangled states of rank four. This construction
gives us explicit formulae for those entangled states which are not obtained by the construction from unextendible product basis.
It would be very nice to know whether this construction is general enough to get all three qubit PPT entangled states of rank four.


\end{document}